\documentclass{article}[12pt]
\usepackage{amsmath, amsfonts,amssymb,graphicx}
\usepackage[usenames]{color}

 \newtheorem{thm}{Theorem}
 \newtheorem{cor}[thm]{Corollary}
 
 \newtheorem{lem}[thm]{Lemma}
 \newtheorem{prop}[thm]{Proposition}
 \newtheorem{assu}[thm]{Assumption}
 
\newtheorem{alg}[thm]{Algorithm}

 \newtheorem{defn}[thm]{Definition}
 \newtheorem{rema}[thm]{Remark}

\newenvironment{rem}%
{  \begin{rema}%
\sl}
 {\end{rema} }
\newenvironment{proof}{{\par\noindent\bf Proof}.}{\hfill $\Box$\\}

 \newcommand{\A}{\mathcal{A}}
 \newcommand{\I}{\mathcal{I}}

 \newcommand{\N}{\mathbb{N}}

 \newcommand{\GBs}{Gr{\"o}bner Bases}

 \def\xb{{\bf x}}
 \def\yb{{\bf y}}
 
\def\CC{\mathbb{K}}
\def\fh{{\bf f}^h}
\def\f{{\bf f}}
\def\Mac{{\rm Mac}_{\Delta}(\mathbf{f})}

\begin{document}


\title{Moment Matrices, Trace Matrices and the Radical of Ideals}

%

\author{
%
I. Janovitz-Freireich\footnote{North Carolina State University, Department
of Mathematics. {\small North Carolina State University},
 {\small Campus Box 8205}, {\small Raleigh, NC, 27695}, {\small USA}
\texttt{\small\{ijanovi2, aszanto\}@ncsu.edu}.
Research supported by NSF
grant CCR-0347506.
}
\ \ \&\ 
A. Sz\'{a}nt\'{o}$^{*}$
\ \&\ 
B. Mourrain\footnote{GALAAD, INRIA, Sophia Antipolis, France, \texttt{\small mourrain@sophia.inria.fr}. Research partially supported by the ANR GECKO.            }
\ \&\ 
L. R\'{o}nyai\footnote{Computer and Automation Institute of the Hungarian Academy of
Sciences, and Budapest University of Technology and Economics.
{\small MTA SZTAKI}, {\small 1111 Budapest,}{\small L\'agym\'anyosi u. 11}, {\small Hungary},
\texttt{\small lajos@csillag.ilab.sztaki.hu}.
Research supported in part by OTKA grant NK63066.
}
}

\maketitle

\begin{abstract}
Let $f_1,\ldots,f_s \in \mathbb{K}[x_1,\ldots,x_m]$ be a system of
polynomials  generating a zero-dimensional ideal $\I$, where $\mathbb{K}$ is an
arbitrary  algebraically closed field. Assume that the factor
algebra $\A=\mathbb{K}[x_1,\ldots,x_m]/\I$ is Gorenstein and that we
have a bound $\delta>0$ such that a basis for $\A$ can be computed
from multiples of $f_1,\ldots,f_s$ of degrees at most $\delta$. We
propose a method using Sylvester or Macaulay type resultant matrices
of  $f_1,\ldots,f_s$ and $J$, where $J$ is a polynomial of degree
$\delta$ generalizing the Jacobian, to compute moment matrices, and
in particular matrices of traces for $\A$.  These matrices of traces
in turn allow us to compute a system of multiplication matrices
$\{M_{x_i}|i=1,\ldots,m\}$ of the radical $\sqrt{\I}$, following the
approach in the previous work by Janovitz-Freireich, R\'{o}nyai and
Sz\'ant\'o. Additionally, we give bounds for $\delta$ for the case
when $\I$ has finitely many projective roots in $\mathbb{P}^m_\CC$.
\end{abstract}




\section{Introduction}

 This paper is a continuation
of our previous investigation in \cite{JaRoSza06, JaRoSza07} to
compute the approximate radical of a zero dimensional ideal which
has zero clusters. The computation of  the radical of a zero
dimensional ideal is a very important problem in computer algebra
since a lot of the algorithms for solving polynomial systems with
finitely many solutions need to start with a radical ideal.  This is
also the case in many numerical approaches, where Newton-like
methods are used. From a symbolic-numeric perspective, when we are
dealing with approximate polynomials, the zero-clusters create great
numerical instability, which can be eliminated by computing the
approximate radical.

The theoretical basis of the symbolic-numeric algorithm presented in
\cite{JaRoSza06, JaRoSza07} was Dickson's lemma \cite{dickson},
which, in the exact case, reduces the problem of computing the
radical of a zero dimensional ideal to the computation of the
nullspace of the so called matrices of traces (see Definition
\ref{mtraces}):  in \cite{JaRoSza06, JaRoSza07} we studied numerical
properties of the matrix of traces when the roots are not multiple
roots, but form small clusters. Among other things we showed that
the direct computation of the matrix of traces (without the
computation of the multiplication matrices) is preferable since the
matrix of traces is continuous with respect to root perturbations
around multiplicities while multiplication matrices are generally
not.

It turns out that the computationally most expensive part of the
method in \cite{JaRoSza06, JaRoSza07} is the computation of the
matrix of traces. We address this problem in the present paper, and
give a simple algorithm using only Sylvester or Macaulay type
resultant matrices and elementary linear algebra to compute matrices
of traces of zero dimensional ideals satisfying certain conditions.

More precisely, we need the following assumptions: let
  ${\bf f}=[f_1,\ldots,f_s]$ be a system of polynomials of degrees
  $d_1\geq\cdots\geq d_s$ in $\CC[{\bf x}]$, with ${\bf
x}=[x_1,\ldots,x_{m}]$, generating an ideal $\I$ in $\CC[\xb]$,
where $\CC$ is an arbitrary  algebraically closed field. We assume
that the algebra  $\A:= \CC[{\bf x}]/\I$ is finite dimensional over
$\CC$ and that we have a   bound $\delta>0$ such that a basis
$S=[b_1, \ldots, b_N]$ of $\A$ can be  obtained
 by taking a linear basis of the space of polynomials of degree at most $\delta$ factored by  the subspace generated by  the multiples of $f_1,\ldots,f_s$  of
degrees at most $\delta$. By slight abuse of notation we denote the
elements of the basis $S$ which are in $\A$ and some fixed preimages of them in
$\CC[{\bf x}]$ both by $b_1,\ldots,b_N$. Thus we can assume that the basis $S$ consists of monomials of degrees at most $\delta$. Note that we can prove bounds
$\delta=\sum_{i=1}^{m+1}d_i-m$  (or $\delta=\sum_{i=1}^{m}d_i-m$ if
$s=m$)  if $\I$ has only finitely many {\it projective} common roots
in $\mathbb{P}^m_\CC$
and have no common roots at infinity,
using a result of Lazard \cite{La81} (see
Theorem \ref{deltatheorem}).

 Furthermore, we also assume that $\A$ is
Gorenstein over $\CC$ (see Definition \ref{gorenstein}). 
Note that in practice we can easily detect if $\A$ is not Gorenstein (see Remark \ref{rem0}).
Also, a random change of projective  variables can eliminate roots at infinity
with high probability when they are in finite number, but we will address the
necessity of this assumption in an upcoming paper.

The main ingredient of our method is a Macaulay type resultant matrix
$\Mac$, which is defined to be a maximal row-independent submatrix of the transpose
matrix  of the degree $\Delta$ Sylvester map $(g_1, \ldots, g_s)\mapsto \sum_{i=1}^s f_ig_i\in \CC[\xb]_\Delta$ for $\Delta\leq 2\delta+1$ (see Definition \ref{bigdelta}). Using our assumptions on $\A$,  we can compute a basis
$S$ of $\A$ using $\Mac$, and   we also prove that a random element $\yb$ of the  nullspace of $\Mac$ provides a
non-singular $N\times N$ moment matrix ${\mathfrak M}_S(\yb)$ with high probability (similarly as in \cite{LaLaRo07}).
This moment matrix allows us to compute the other main ingredient of our algorithm, a polynomial $J$ of degree at most
$\delta$, such that $J$ is the generalization of the Jacobian of $f_1, \ldots f_s$ in the case when $s=m$. The main result
 of the paper now can be formulated as follows: \\

{\sc Theorem} {\it Let $S=[b_1, \ldots, b_N]$ be a basis of $\A$
with $\deg(b_i)\leq \delta$. With $J$ as above, let ${\rm Syl}_S(J)$
be the transpose matrix  of the map $\sum_{i=1}^N c_i b_i\mapsto
J\cdot\sum_{i=1}^N c_i b_i\;\in \CC[x]_\Delta$ for $c_i\in \CC$.
Then $$ \left[Tr(b_ib_j)\right]_{i,j=1}^N= {\rm Syl}_S(J)\cdot X,
$$
where $X$ is the unique extension of the matrix $\mathfrak{M}_S({\bf
y})$ such that
$\Mac \cdot X=0.$}\\

Once we compute the matrix of traces
$R:=\left[Tr(b_ib_j)\right]_{i,j=1}^N$ and the matrices
$R_{x_k}:=\left[Tr(x_kb_ib_j)\right]_{i,j=1}^N= {\rm
Syl}_S(x_kJ)\cdot X$ for $k=1, \ldots, m$, we can use the results of
\cite{JaRoSza06, JaRoSza07} to compute a system of multiplication
matrices for the (approximate) radical of $\I$ as follows: if
$\tilde{R}$ is a (numerical) maximal non-singular submatrix of $R$
and $\tilde{R}_{x_k}$ is the submatrix of $R_{x_k}$ with the same
row and column indices as in $\tilde{R}$, then  the solution
$M_{x_k}$ of  the linear matrix equation
$$\tilde{R}M_{x_k}=\tilde{R}_{x_k}$$
is an (approximate) multiplication matrix of $x_k$  for the
(approximate) radical of $\I$.
See \cite{JaRoSza07} for the definition of (approximate) multiplication matrices. Note that a generating set for the radical $\sqrt{\I}$ can be obtained directly from the definition
of multiplication matrices, in particular, it corresponds to the rows of the matrices $M_{x_1}, \ldots,
M_{x_m}$.

We also point out that in the $s=m$
case these multiplication matrices $M_{x_k}$ can be obtained even more
simply using the nullspace of $\Mac$ and the Jacobian $J$ of ${\bf
f}$, without computing the matrices of traces.

We also note here that in a follow up paper we will consider an
extension of our present results which works also in the
non-Gorenstein case to compute the matrices of traces. Furthermore,
that paper will also extend our results to the {\it affine} complete
intersection case using Bezout  matrices.

\section{Related Work}

The motivation for this work was the papers \cite{LaLaRo07,
LaLaRo072} where they use moment matrices to compute the radical of
real and complex ideals. They present two versions of the method for
the complex case: first, in \cite{LaLaRo072} they double up the
machinery for the real case to obtain the radical of the complex
ideal.   However, in \cite{LaLaRo07} they significantly simplify
their method and show how to use moment matrices of maximal rank to
compute the multiplication matrices of an ideal between $\I$ and its
radical $\sqrt{\I}$. In particular, in the Gorenstein case they can
compute the multiplication matrices of $\I$. In fact, in
$\cite{LaLaRo07}$ they cite our previous work \cite{JaRoSza06} to
compute the multiplication matrices of $\sqrt{\I}$  from  the multiplication matrices of
$\I$, but the method proposed in the present paper is much simpler
and more direct.

Note that one can also obtain the multiplication matrices  of $\I$
with respect to the basis $S=[b_1, \ldots, b_N]$ by simply
eliminating the terms not in $S$ from $x_kb_i$ using ${\rm
Mac}_{\delta+1}({\bf f})$. The advantage of computing multiplication
matrices of the radical $\sqrt{\I}$  is that it  returns matrices
which are always simultaneously diagonalizable, and possibly smaller
than the multiplication matrices of $\I$, hence easier to work with. Moreover, if $S$ contains the
monomials $1,x_{1},\ldots,x_{m}$, one eigenvector computation yield directly
the coordinates of the roots.

Computation of  the radical of zero dimensional complex ideals is
very well studied in the literature:  methods most related to ours
include  \cite{GonVeg, becwor96} where matrices of traces are used
in order to find  generators of the radical, and the matrices of
traces are computed using \GBs; also, in \cite{ArSo95} they use
the traces to give a bound for the degree of the generators of the
radical and use linear solving methods from there; in \cite{GoTr95}
they describe the computation of the radical using symmetric
functions which are related to traces. One of the most commonly
quoted method to compute radicals is to compute the  projections
$\I\cap\CC[x_i]$ for each $i=1, \ldots, m$ and then use univariate
squarefree factorization (see for example
\cite{GiTrZa88,KrLo91b,Cox98,GrPf02} ). The advantage of the latter
is that it can be generalized for higher dimensional ideals (see for
example \cite{KrLo91}). We note here that an advantage of the method
using matrices of traces is that it behaves stably under
perturbation of the  roots of the input system, as was proved in
\cite{JaRoSza07}. Other methods to compute the radical of zero
dimensional ideals include \cite{KoMoHo89,GiMo89,
Lak90,Lak91,LakLaz91,YoNoTa92}. Applications of computing the
radical include \cite{HeObPa06}, where they show how to compute the
multiplicity structure of the roots of $\I$ once the radical is
computed.

Methods for computing the matrix of traces directly from the
generating polynomials of $\I$, without using multiplication
matrices, include \cite{DiazGonz01,brigonz02} where they use Newton
Sums, \cite{carmou96,CatDicStu96,CatDicStu98} where they use
residues and \cite{DaJe05} using resultants. Besides computing the
radical of an ideal,  matrices of traces have numerous applications
mainly in real algebraic geometry \cite{Be91, PeRoSz93, BeWo94}, or
in \cite{Rouiller99}  where trace matrices are applied  to find
separating linear forms deterministically.

\section{Moment Matrices and Matrices of Traces}

Let  ${\bf f}=[f_1,\ldots,f_s]$ be a system of polynomials of
degrees $d_1\geq\cdots\geq d_s$ in $\mathbb{K}[{\bf x}]$, where
${\bf x}=[x_1,\ldots,x_{m}]$ and $\CC$ is an arbitrary algebraically
closed field. Let $\I$ be the ideal generated by $f_1,\ldots,f_s$ in
$\mathbb{K}[{\bf x}]$ and define $\A:=\mathbb{K}[{\bf x}]/\I$.  We
assume throughout the paper that $\A$  is a finite dimensional
vector space over $\CC$  and let $\A^*$ denote the dual space of
$\A$.

Let us first recall the definition of a Gorenstein algebra (c.f.
\cite{Kun86,ScSt75,ElMo07,LaLaRo07}). Note that these algebras are also referred to as Frobenius in the literature, see for example \cite{BeCaRoSz96}.

\begin{defn}\label{gorenstein}
A finite dimensional $\CC$-algebra $\A$ is Gorenstein (over $\CC$)
if there exists a nondegenerate $\CC$-bilinear form $B(x,y)$ on $\A$
such that
\[
B(ab,c)=B(a,bc) \; \text{ for every } a,b,c \in \A.
\]
\end{defn}

Note that this is equivalent to the fact that $\A$ and $\A^*$ are
isomorphic as $\A$ modules. It is also equivalent to the existence
of a ${\mathbb K}$-linear function $\Lambda: \A\rightarrow {\mathbb
K}$ such that the bilinear form $B(a,b):=\Lambda(ab)$ is
nondegenerate on $\A$.


\begin{assu}\label{assump}
Throughout the paper we assume that $\A$ is Gorenstein. Furthermore,
we also assume that we have a bound $\delta>0$ such that
\begin{eqnarray}\label{deltacond}
N:=\dim_\mathbb{K}\mathbb{K}[\xb]_{\delta}/\langle f_1, \ldots,
f_s\rangle_{\delta}=\dim_\mathbb{K}\mathbb{K}[\xb]_{d}/\langle f_1,
\ldots, f_s\rangle_{d}
\end{eqnarray} for all $d\geq \delta$ and that
\begin{eqnarray}\label{dimA}
N=\dim \A.
 \end{eqnarray}
 Here
\begin{eqnarray}\label{fd}
\langle f_1, \ldots, f_s\rangle_{d}:=\left\{\sum_i f_i p_i\; :\;
   \deg(p_i) \leq d-d_i\right\}.
 \end{eqnarray}
 We fix $S=[b_1,\ldots,b_N]$ a monomial basis for $\A$ such that ${\rm deg}( b_i)\leq\delta$ for all $i=1,\ldots,N$. Let $D$ be the maximum degree of the monomials in $S$. Thus $D\leq \delta$.
\end{assu}

We have the following theorem giving bounds for $\delta$ in the case when ${\bf f}$ has finitely many projective roots.

\begin{thm}\label{deltatheorem}
Let  ${\bf f}=[f_1,\ldots,f_s]$ be a system of polynomials of
degrees $d_1\geq\cdots\geq d_s$ in $\mathbb{K}[{\bf x}]$. Assume
that $f_1,\ldots,f_s$ has finitely many projective common roots in
$\mathbb{P}_\CC^m$. Assume further that $f_1,f_2, \ldots ,f_s$ have no common roots at infinity. Then:
\begin{enumerate}

\item If $s=m$  then for $\delta:=\sum_{i=1}^m (d_i-1)$
    conditions (\ref{deltacond}) and (\ref{dimA}) are satisfied. Furthermore, in this case $\A$ is always Gorenstein.

\item If $s> m$ then
for $\delta:=\sum_{i=1}^{m+1} d_i-m$ conditions (\ref{deltacond}) and
(\ref{dimA}) are  satisfied.
\end{enumerate}
 \end{thm}

\begin{proof}

For the first assertion let $\fh$ be the homogenization  of ${\bf f}$ using a new variable $x_{m+1}$. Using our assumption that $\fh$  has finitely many roots in
$\mathbb{P}_{\mathbb{K}}^m$ and $s=m$, one can see that $(\fh)$ is a
regular sequence in $R:={\mathbb K}[x_1, \ldots, x_m,x_{m+1}]$.

Define the graded ring  $B:=R/\langle \fh\rangle$. Following the approach
and notation
 in
\cite{Stan96}, we can now calculate the Hilbert series of $B$, defined by
$H(B,\lambda)=\sum_d {\cal H}_B(d) \lambda^d$, where $\cal{H}_B$ is the Hilbert function of $B$. We
have
$$H(R,\lambda)=\frac{H(B,\lambda)}{(1-\lambda^{d_1})\cdots (1-\lambda^{d_m})},$$
and using the simple fact that
$$H(R, \lambda)=\frac{1}{(1-\lambda)^{m+1}}$$
we obtain that
$$
\begin{aligned}
 H(B,\lambda)&=\frac{(1+\lambda+\cdots +\lambda^{d_1-1})\cdots
(1+\lambda+\cdots
+\lambda^{d_m-1})}{(1-\lambda)}\notag\\
&=g(\lambda)(1+\lambda+\ldots),\notag
 \end{aligned}$$
 where
$$g(\lambda) =(1+\lambda+\cdots +\lambda^{d_1-1})\cdots (1+\lambda+\cdots
+\lambda^{d_m-1}).$$
This implies that the Hilbert function
$${\cal H}_B(\delta)={\cal H}_B(\delta+1)={\cal H}_B(\delta+2)=\ldots $$
Note that dehomogenization induces a linear
isomorphism $B_d \rightarrow {\mathbb K}[{\bf x}]_d/\langle f_1,
\ldots, f_s\rangle_{d}$, where $B_d$ stands for the degree $d$
homogeneous part of $B$. From this, using that there are no common roots at infinity,
we infer that for $d\geq \delta$
$\dim _{\mathbb K}{\mathbb K}[{\bf x}]_d/\langle f_1,
\ldots, f_s\rangle_{d}=\dim_\CC \A = N$, which implies  (\ref{deltacond}) and (\ref{dimA}).

Note that the common value $N={\cal H}_B(\delta)$ is the sum of the coefficients
of $g$, which is
$$g(1)=\prod_{i=1}^m d_i.$$

To prove that $\A$ is Gorenstein, we cite  \cite[Proposition 8.25, p. 221]{ElMo07} where it is proved that if $f_{1},\ldots, f_{m}$ is an affine
complete intersection then the Bezoutian  $B_{1,f_{1},\ldots,f_{m}}$
defines an isomorphism between ${\A}^*$ and $\A$.

To prove the second assertion we note that \cite[Theorem 3.3]{La81} implies that
$$
\dim _{\mathbb K}B_\delta=\dim _{\mathbb K}B_{\delta+1}=\ldots .
$$
From here we obtain (\ref{deltacond}) and (\ref{dimA}) as in the Case 1.

\end{proof}

\begin{rem} Note that in general $\I_d\neq\langle f_1,\ldots,f_s \rangle_d$, where $I_d$ is the set of elements of $I$ with degree at most $d$ and $\langle f_1,\ldots,f_s \rangle_d$ was defined in (\ref{fd}). This can happen when the system has a root at infinity, for example, if $f_1=x+1, \;f_2=x$ then
$\I_0={\rm span}_\CC(1)$ but $\langle f_1,f_2 \rangle_0=\{0\}$
However, using the  homogenization $f_1^h,\ldots,f_s^h$, the degree
$d$ part of the homogenized ideal is always equal to the space
spanned by the multiples of $f_1^h,\ldots,f_s^h$ of  degree $d$. The above example also demonstrates
that $\dim \A$ is not always the same as $\dim\CC[\xb]_d/\langle f_1,\ldots,f_s \rangle_d$ for large enough $d$, because above $\dim\A=0$ but  $\dim\CC[x,y]_d/\langle f_1,f_2 \rangle_d=1$ for all $d\geq 0$.
\end{rem}

Next we will define Sylvester and Macaulay type resultant matrices
for $f_1,\ldots f_s$.

\begin{defn}\label{bigdelta}
Define $$\Delta:=\max(\delta,2D+1)$$ where $\delta$ and $D$ are
defined in Assumption \ref{assump}. 

Let ${\rm Syl}_{\Delta}(\mathbf{f})$ be the transpose matrix of the
linear map
\begin{eqnarray}
\bigoplus_i \CC[{\bf x}]_{\Delta-d_i}&\longrightarrow &\CC[{\bf x}]_{\Delta}\label{syl}\\
(g_1,\ldots,g_s)&\mapsto& \sum_{i=1}^sf_ig_i\notag
\end{eqnarray}
written in the monomial bases. So, in our notation, ${\rm
Syl}_{\Delta}(\mathbf{f})$ will have rows which correspond to all polynomials $f_ix^\alpha$ of degree at most
   $\Delta$.

Let $\Mac$ be a row submatrix of ${\rm Syl}_{\Delta}({\bf f})$ of maximal size with
linearly independent rows.
\end{defn}

\begin{rem} In the case where $s=m$, for generic ${\bf f}$ we can directly construct $\Mac$ by taking the restriction of the map (\ref{syl}) to
\[
\bigoplus_{i=1}^m {\cal S}_i(\Delta)\longrightarrow
\CC[{\bf x}]_{\Delta}
\]
where ${\cal S}_i(\Delta)={\rm span}\{{\bf x}^{\alpha}:|\alpha|\leq\Delta-d_i,\,
\forall
 j< i,\,\alpha_j< d_j\}$.

Here $\Mac$ is a submatrix of the classical Macaulay matrix of the
homogenization of $\f$ and some $f^h_{m+1}$, where $f^h_{m+1}$ is
any homogeneous polynomial of degree $\Delta-\delta$: we only take
the rows corresponding to the polynomials in $\f$. Since the
Macaulay matrix is generically non-singular, $\Mac$ will also be
generically full rank.

Note that with our assumption that $f_1,\ldots,f_m$ has finitely
many projective roots, we have that $\Mac$ has column \rm{corank} $
N:=\prod_{i=1}^m d_i. $

\end{rem}

Since $\Delta\geq\delta$, by Assumption \ref{assump} the corank of
$\Mac=N$, where $N$ is the dimension of $\A$. Also, we can assume
that the elements of the basis $S$ of $\A$ are monomials of degree
at most $\delta$, and   that the first columns of $\Mac$ correspond
to the  basis $S$ of $\A$.

Fix an element $${\bf y}=[y_{\alpha}:\alpha\in \N^m, \;|\alpha|\leq
\Delta]^T$$ of the nullspace ${\rm Null}(\Mac)$, i.e. $\Mac\cdot\yb=0$.

\begin{defn}\label{moment}
Let $S$ be the basis of $\A$ as above, consisting of monomials of degree at most $D$.
 Using ${\bf y}$ we can  define $\Lambda_\yb\in \A^*$ by
$
\Lambda_\yb (g):=\sum_{\xb^\alpha\in S} y_{\alpha}g_{\alpha},
$
where $g=\sum_{\xb^\alpha\in S}g_{\alpha}\xb^\alpha\in \A$. Note
that every $\Lambda\in \A^*$ can be defined as $\Lambda_\yb$ for
some $\yb\in {\rm Null}(\Mac)$ or more generally with an element of
${\mathbb K}[{\bf x}]^{*}$ which vanishes on the ideal
$\I$.

Define the \emph{moment matrix} $\mathfrak{M}_S({\bf y})$ to be the
$N\times N$ matrix given by
$$
 \mathfrak{M}_S({\bf y})=[y_{\alpha+\beta}]_{\alpha,\beta},$$
 where $\alpha$ and $\beta$ run through the exponents of the monomials in $S$. Note that $\mathfrak{M}_S$ is only a
   submatrix of the usual notion of moment matrices in the literature, see for example
   \cite{CuFi96}.

For $p \in \A$, we define the linear function  $p\cdot \Lambda\in
\A^*$ as $p\cdot \Lambda(g):=\Lambda(pg)$ for all $g\in \A$.

\end{defn}

\begin{rem}
If one considers a linear function $\Lambda$ on $\A$, such that the
bilinear form $(x,y)\mapsto\Lambda(xy)$ is nondegenerate on $\A$,
then the moment matrix corresponding to this $\Lambda$ will be the
one whose $(i,j)$-th entry  is just $\Lambda(b_ib_j)$. Moreover, for
$g,h\in \A$
\[
\Lambda_\yb(gh)={\rm coeff}_S(g)^T\cdot\mathfrak{M}_S(\yb)\cdot{\rm
coeff}_S(h)
\]
where ${\rm coeff}_S(p)$ denotes the vector of coefficients of $p\in
\A$ in the basis $S$.
\end{rem}

The following proposition is a simple corollary of \cite[Prop 3.3
   and Cor. 3.1]{LaLaRo07}.

\begin{prop}
Let $\,{\bf y}$ be a random element of the vector space ${\rm
Null}(\Mac)$. With high probability, $\mathfrak{M}_S({\bf y})$ is
non-singular.
\end{prop}


 \begin{rem}\label{rem0}
Using the above proposition, one can detect whether the algebra $\A$ is not Gorenstein with high probability by simply computing the rank of $\mathfrak{M}_S({\bf y})$ for (perhaps several) random elements $\,{\bf y}$ in ${\rm
Null}(\Mac)$.
\end{rem}

\begin{rem}\label{rem1}
%
%
By \cite[Theorem 2.6 and Lemma
3.2]{LaLaRo07} one can extend ${\bf y}$ to $\tilde{\bf
y}\in\CC^{\N^m}$ such that the infinite moment matrix
$\mathfrak{M}(\tilde{\bf y}):=[\tilde{ y}_{\alpha+\beta}]_{\alpha,
\beta\in \N^m}$ has the same rank as $\mathfrak{M}_S({\bf y})$ and
the columns of $\mathfrak{M}(\tilde{\bf y})$ vanish on all the
elements of the ideal $\I$.
\end{rem}

Next we define a basis dual to $S=[b_1,\ldots,b_N]$ with respect to
the moment matrix $\mathfrak{M}_S({\bf y})$.  Using this dual basis
we also define a polynomial $J$ which is in some sense a
generalization of the Jacobian of a well-constrained polynomial
system.

\begin{defn} \label{dualbasis}

From now on we fix $\yb\in {\rm Null}(\Mac)$ such that
$\mathfrak{M}_S({\bf y})$ is invertible and we will denote by
$\Lambda$ the corresponding element $\Lambda_\yb\in \A^*$. We define

$$\mathfrak{M}_S^{-1}({\bf y})=:[c_{ij}]_{i,j=1}^N.$$

Let $b_i^*:=\sum_{j=1}^N c_{ji}b_j.$
Then $[b_1^*, \ldots, b_N^*]$ corresponds to the columns of the
inverse matrix  $\mathfrak{M}_S^{-1}({\bf y})$ and   they also form a basis
for $\A$. Note that we have $\Lambda(b_ib_j^*)=1$, if $i=j$, and 0
otherwise.

Define the {\em generalized Jacobian} by
\begin{eqnarray}
J:=\sum_{i=1}^Nb_ib_i^*\;\text{ mod }  \I\label{J}
\end{eqnarray}
expressed in the basis $S=[b_1, \ldots, b_N]$ of $\A$.

\end{defn}

\begin{rem}
Note that since $\sum_{i=1}^N b_ib^*_i$ has degree at most $2D$, and $\Delta>2D$,
we can use $\Mac$ to find its reduced form, which is $J$. Because of
this reduction, we have that ${\rm deg}(J)\leq D\leq \delta$.

Also note that the notion of generalized Jacobian   was also introduced in
\cite{BeCaRoSz96}. Its name come from the fact that if $\,s=m$ and if $\,\Lambda$ is the so called
residue (c.f. \cite{ElMo07}),  then $\sum_{i=1}^N b_ib^*_i=J$ is the
Jacobian of $f_1, \ldots, f_m$. 
\end{rem}

We now recall the definition of the multiplication matrices and the
matrix of traces as presented in \cite{JaRoSza07}.

\begin{defn}\label{mtraces}
Let $p\in \A$. The multiplication matrix $M_p$ is the transpose of
the matrix of the multiplication map
\[
\begin{aligned}
M_p:\A&\longrightarrow \A\notag\\
g&\mapsto pg
\end{aligned}
\]
written in the basis $S$.

The \emph{matrix of traces} is the $N\times N$ symmetric matrix:
\[
R=\left[Tr(b_ib_j)\right]_{i,j=1}^N
\]
where $Tr(pq):= Tr(M_{pq})$,  $M_{pq}$ is the multiplication matrix
of $pq$ as an element in $\A$ in terms of the basis $S=[b_1, \ldots,
b_N]$ and $Tr$ indicates the trace of a matrix.
\end{defn}

The next results relate the multiplication by $J$ matrix to the
matrix of traces $R$.

\begin{prop}
Let $M_J$ be the multiplication matrix of $J$ with respect to the
basis $S$. We then have that
$$M_J=[Tr(b_ib_j^*)]_{i,j=1}^N. $$
\end{prop}

\begin{proof}
Let $\Lambda\in \A^*$ be as in Definition \ref{dualbasis}. For any
$h\in \A$ we have that
\[
\begin{aligned}
&h=\sum_{j=1}^N \Lambda(hb_j)b_j^*=\sum_{j=1}^N
\Lambda(hb_j^*)b_j\notag\\
\Rightarrow\quad &hb_i=\sum_{j=1}^N \Lambda(hb_j^*b_i)b_j= M_h[i,j]=\Lambda(hb_j^*b_i)\notag\\
\Rightarrow\quad &Tr(h)=\sum_{i=1}^N
\Lambda(hb_i^*b_i)=\Lambda(h\sum_{i=1}^N b_i^*b_i).\notag
\end{aligned}
\]

Since $J=\sum_{i=1}^N b_i^*b_i$ in $\A$, we have $Tr(h)=\Lambda(hJ)$.

Therefore
\[
M_J[i,j]=\Lambda(Jb_j^*b_i)=Tr(b_j^*b_i) =Tr(b_ib_j^*)\]
\end{proof}

\begin{cor}
\[
M_J\cdot\mathfrak{M}_S({\bf y})=[Tr(b_ib_j)]_{i,j=1}^N=R,
\]
or equivalently $J\cdot \Lambda = Tr$ in $\A^{*}$.
\end{cor}

\begin{proof}
The coefficients of $b_i^*$ in the basis $S=[b_1,\ldots,b_N]$
are the columns of $\mathfrak{M}_S^{-1}({\bf y})$, which implies
that
\[
M_J=[Tr(b_ib_j^*)]_{i,j=1}^N=[Tr(b_ib_j)]_{i,j=1}^N\cdot\mathfrak{M}_S^{-1}({\bf
y}).
\]
Therefore we have that $M_J\cdot\mathfrak{M}_S({\bf
y})=[Tr(b_ib_j)]_{i,j=1}^N$.
\end{proof}

Finally, we prove that the matrix of traces $R$ can be computed
directly from the Sylvester matrix of $f_1, \ldots, f_s$ and $J$,
without using the multiplication matrix $M_J$. First we need a
lemma.

\begin{lem}\label{Rlemma}
There exists a unique matrix $\mathfrak{R}_S({\bf y})$ of size
$|{\rm Mon}_{\leq}(\Delta)-S| \times |S|$ such that
\begin{small}
\[
\Mac\cdot\begin{tabular}{|c|} \hline
\\
$\mathfrak{M}_S({\bf y})$\\
\\
\hline
\\
$\mathfrak{R}_S({\bf y})$\\
\\
\hline
\end{tabular}=0
\]
\end{small}
\end{lem}

\begin{proof}
By our assumption that the first columns  of ${\rm
Mac}_{\Delta}(\f)$ correspond to $S$ we have
\[
\begin{small}
\begin{tabular}{|ccc|}
\hline
&&\\
&$\Mac$&\\
&&\\
\hline
\end{tabular}=\begin{tabular}{|ccc|ccc|}
 \hline
&&&&&\\
&$B$&&&$A$&\\
&&&&&\\
\hline
\end{tabular}\,,
\end{small}
\]
where the columns of $B$ are indexed by the monomials in $S$.
Note here that  by
Assumption \ref{assump} the rows of $\Mac$ span $\I_\Delta$, and the monomials in $S$ span the factor
space $\CC[\xb]_\Delta/\I_\Delta$. These together imply that the (square)
submatrix $A$ is invertible.

Then
\begin{small}
\[\begin{tabular}{|ccc|ccc|}
 \hline
&&&&&\\
&$B$&&&$A$&\\
&&&&&\\
\hline
\end{tabular}\cdot
\begin{tabular}{|c|}
\hline
\\
$Id_{N\times N}$\\
\\
\hline
\\
$-A^{-1}B$\\
\\
\hline
\end{tabular}=0
\]
\end{small}

which implies that
\begin{small}
\[
\Mac\cdot\begin{tabular}{|c|} \hline
\\
$\mathfrak{M}_S({\bf y})$\\
\\
\hline
\\
$\mathfrak{R}_S({\bf y})$\\
\\
\hline
\end{tabular}=0,
\]
\end{small}
where $\mathfrak{R}_S({\bf y})=-A^{-1}B\cdot\mathfrak{M}_S({\bf
y})$.

\end{proof}

By construction, the column of $\mathfrak{M}_S({\bf y})$ indexed by
$b_{j} \in S$ corresponds to the values of $b_{j}\cdot \Lambda \in
\A^{*}$ on $b_{1},\ldots,b_{N}$. The same column in
$\mathfrak{R}_S({\bf y})$ corresponds to the values of $b_{j}\cdot
\Lambda$ on the complementary set of monomials of
$\mathrm{Mon}_{\le}(\Delta)$. The column in the stacked matrix
corresponds to the value of $b_{j}\cdot \Lambda$ on all the
monomials in $\mathrm{Mon}_{\le}(\Delta)$.  To evaluate $b_{j}\cdot
\Lambda(p)$ for a polynomial $p$ of degree $\le \Delta$, we simply
compute the inner product of the coefficient vector of $p$ with this
column.

\begin{defn}\label{SylS}
Let $S=[b_1, \ldots, b_N]$ be the basis of $\A$ as above, and let
$P\in \CC[\xb]$ be a polynomial of degree at most $D+1$.

Define ${\rm Syl}_{S}(P)$ to be the matrix with rows corresponding
to the coefficients of the polynomials $(b_1P),\ldots,(b_NP)$ in the
monomial basis ${\rm Mon}_{\leq}(\Delta)$ (we use here that ${\rm
deg}(b_i)\leq D$, thus ${\rm deg}(b_iP)\leq 2D+1\leq\Delta$).

Furthermore,  we assume that the monomials corresponding to the
columns of ${\rm Syl}_{S}(P)$ are  in the same order as the
monomials  corresponding to  the columns of $\Mac$.
\end{defn}

\begin{thm}
\begin{small}
\[
\begin{tabular}{|ccc|}
\hline
&&\\
&${\rm Syl}_{S}(J)$&\\
&&\\
\hline
\end{tabular}\cdot
\begin{tabular}{|c|}
\hline
\\
$\mathfrak{M}_S({\bf y})$\\
\\
\hline
\\
$\mathfrak{R}_S({\bf y})$\\
\\
\hline
\end{tabular}=[Tr(b_ib_j)]_{i,j=1}^N
\]
\end{small}
\end{thm}

\begin{proof}
Since the $j$-th column of the matrix
\begin{small}
$$
\begin{tabular}{|c|}
\hline
\\
$\mathfrak{M}_S({\bf y})$\\
\\
\hline
\\
$\mathfrak{R}_S({\bf y})$\\
\\
\hline
\end{tabular}
$$
\end{small}
represents the values of $b_j\cdot \Lambda$ on all the monomials of
degree less than or equal to $\Delta$, and the $i$-th row of ${\rm Syl}_{S}(J)$ is the coefficient matrix of
$b_i J$, we have
\begin{small}
\[
\begin{aligned}
\begin{tabular}{|ccc|}
\hline
&&\\
&${\rm Syl}_{S}(J)$&\\
&&\\
\hline
\end{tabular}\cdot
\begin{tabular}{|c|}
\hline
\\
$\mathfrak{M}_S({\bf y})$\\
\\
\hline
\\
$\mathfrak{R}_S({\bf y})$\\
\\
\hline
\end{tabular}&=\left[(b_j \cdot \Lambda)(b_i J)\right]_{i,j=1}^N\notag\\
&=\left[\Lambda(J b_i b_j)\right]_{i,j=1}^N\notag\\
&=[Tr(b_ib_j)]_{i,j=1}^N.\notag
\end{aligned}
\]
\end{small}
\end{proof}

We can now describe the algorithm to compute a set of multiplication
matrices $M_{x_i}$, $i=1,\ldots,m$ of the radical $\sqrt{\I}$ of
$\I$ with respect to a  basis of $\CC[\xb]/\sqrt{\I}$.
To prove that
the algorithm below is correct we need the following result from
\cite[Proposition 8.3]{JaRoSza07}
which is the consequence of the  fact that the kernel of the matrix of traces corresponds to the radical of  $\A$:

\begin{prop}
Let $\tilde{R}$ be a maximal non-singular submatrix of the matrix of
traces $R$. Let $r$ be the rank of $\tilde{R}$, and $T:=[b_{i_1},
\ldots, b_{i_r}]$ be the monomials corresponding to the columns of
$\tilde{R}$. Then $\,T$ is a basis of the algebra $\CC[{\bf
x}]/\sqrt{\I} $ and for each $k=1, \ldots, m$, the solution
$M_{x_k}$ of  the linear matrix equation
$$\tilde{R}M_{x_k}=\tilde{R}_{x_k}$$
is the  multiplication matrix of $x_k$  for $\sqrt{\I} $ with
respect to  $T$. Here $\tilde{R}_{x_k}$ is the $r\times r$ submatrix
of $[Tr(x_kb_ib_j)]_{i,j=1}^N$ with the same row and column indices
as in $\tilde{R}$.
\end{prop}

\noindent\begin{alg}\label{algo}

\noindent\textsc{Input}: $\f=[f_1,\ldots,f_s]\in\CC[{\bf x}]$
of degrees $d_1,\ldots,d_s$ generating an ideal $\I$ and
$\delta>0$ such
that they satisfy the conditions in Assumption  \ref{assump}. An optional input is $D\leq\delta$, which by default is set to be $\delta$.\\

\noindent\textsc{Output}: A basis $T$ for the factor algebra
$\CC[{\bf x}]/\sqrt{\I}$ and a set of multiplication matrices
$\{M_{x_i}|i=1,\ldots,m\}$ of $\sqrt{\I}$  with respect to the basis
$T$.

\begin{enumerate}

\item Compute ${\rm Mac}_{\Delta}(\f)$ for $\Delta:=\max(2D+1, \delta)$

\item Compute a basis $S$ of $\mathbb{K}[{\bf x}]_{\Delta}/\langle {\bf f}\rangle_{\Delta}$ such that the polynomials in $S$ have degrees at most $D$. Let $S=[b_1,\ldots,b_N]$.

\item Compute a random combination ${\bf y}$
of the elements of a basis of $Null(\Mac)$.

\item Compute the moment matrix $\mathfrak{M}_S({\bf y})$ defined in Definition \ref{moment} and
$\mathfrak{R}_S({\bf y})$ defined in Lemma \ref{Rlemma}.

\item Compute $\mathfrak{M}_S^{-1}({\bf y})$ and the
basis $[b_1^*,\ldots,b_N^*]$ defined in Definition \ref{dualbasis}.

\item Compute $J=\sum_{i=1}^N b_ib_i^*\;\text{ mod }  \I$ using $\Mac$.

\item Compute ${\rm Syl}_{S}(J)$ and ${\rm Syl}_{S}(x_kJ)$ for $k=1,\ldots,m$ defined in Definition \ref{SylS}.

\item Compute \\
\begin{small}$R=[Tr(b_ib_j)]_{i,j=1}^N=$ \begin{tabular}{|ccc|}
\hline
&&\\
&${\rm Syl}_{S}(J)$&\\
&&\\
\hline
\end{tabular} $\cdot$ \begin{tabular}{|c|}
\hline
\\
$\mathfrak{M}_S({\bf y})$\\
\\
\hline
\\
$\mathfrak{R}_S({\bf y})$\\
\\
\hline
\end{tabular}\\

and\\
$R_{x_k}$:=$[Tr(x_kb_ib_j)]_{i,j=1}^N=$ \begin{tabular}{|ccc|}
\hline
&&\\
&${\rm Syl}_{S}(x_kJ)$&\\
&&\\
\hline
\end{tabular} $\cdot$ \begin{tabular}{|c|}
\hline
\\
$\mathfrak{M}_S({\bf y})$\\
\\
\hline
\\
$\mathfrak{R}_S({\bf y})$\\
\\
\hline
\end{tabular}\quad for $k=1,\ldots,m$.
\end{small}

\item Compute $\tilde{R}$,  a maximal non-singular submatrix of
$R$. Let $r$ be the rank of $\tilde{R}$, and $T:=[b_{i_1}, \ldots,
b_{i_r}]$ be the monomials corresponding to the columns of
$\tilde{R}$.

\item For each $k=1, \ldots, m$ solve the linear matrix equation $\tilde{R}M_{x_k}=\tilde{R}_{x_k}$, where
$\tilde{R}_{x_k}$ is the submatrix of ${R}_{x_k}$ with the same row
and column indices as in $\tilde{R}$.

\end{enumerate}
\end{alg}

\begin{rem}  Since the bound given in Theorem \ref{deltatheorem} might be
   too high, it seems reasonable to design the algorithm in
   an iterative fashion, similarly to the algorithms in \cite{LaLaRo07,LaLaRo072,ReZh04}, in order to
   avoid nullspace computations for large matrices.  The bottleneck of our algorithm is doing computations with ${\rm Mac}_{\Delta}(\f)$, since its size exponentially increases as $\Delta$ increases. 
\end{rem}

\begin{rem}
Note that if $s=m$ then we can use the conventional Jacobian of $f_1, \ldots,
f_m$ in the place of $J$, and any $|{\rm Mon}_\leq (\Delta)|\times
|S|$ matrix $X$ such that it has full rank and $\Mac \cdot X={\bf
0}$ in the place of
\begin{small}
\[\begin{tabular}{|c|} \hline
\\
$\mathfrak{M}_S({\bf y})$\\
\\
\hline
\\
$\mathfrak{R}_S({\bf y})$\\
\\
\hline\end{tabular}\,.\]
\end{small}Even though this way we will not get
matrices of traces,  a system of multiplication matrices of the
radical $\sqrt{\I}$ can still be recovered:
 if $\tilde{Q}$ denotes a maximal non-singular submatrix of ${\rm Syl}_S(J)\cdot X$, and
$\tilde{Q}_{x_k}$ is the submatrix of ${\rm Syl}_S(x_k J)\cdot X$
with the same row and column indices as in $\tilde{Q}$, then the
solution $M_{x_k}$ of the linear matrix equation
$\tilde{Q}M_{x_k}=\tilde{Q}_{x_k}$ gives the same multiplication
matrix of $\sqrt{\I}$ w.r.t. the same basis $T$ as the above
Algorithm.
\end{rem}

\begin{rem}
As $M_{x_{k}}$ is the matrix of multiplication by $x_{k}$ modulo the
radical ideal $\sqrt{\I}$, its eigenvectors are (up to a non-zero
scalar) the interpolation polynomials  at the roots of $\I$.
Similarly the eigenvectors of the transposed matrix $M_{x_{k}}^{t}$
are  (up to a non-zero scalar) the evaluation at the roots $\zeta$
of $\I$ (see \cite{BMr98,ElMo07} for more details). The vector which
represents this evaluation at $\zeta$ in the dual space $\A^{*}$ is
the vector of values of $[b_{1}, \ldots, b_{N}]$ at $\zeta$. To
obtain these vectors, we solve the generalized eigenvalue problem
$(\tilde{R}_{x_{k}}^{t}-z \tilde{R}^{t})\, w =0$ and compute $v =
\tilde{R}^{t}\, w$. The vectors $v$ will be of the form
$[b_{1}(\zeta),\ldots, b_{N}(\zeta)]$ for $\zeta$ a root of $\I$. If
$b_{1}=1,b_{2}=x_{1}, \ldots, b_{m+1}=x_{m}$, we can read directly
the coordinates of $\zeta$ from this vector.
\end{rem}

\section{Examples}

In this section we present three examples. Each of them has three
polynomials in two variables. The first one is a system which has
roots with multiplicities, the second one is a system which has clusters
of roots, and the third one is a system obtained by perturbing the
coefficients of the first one. For each of them we compute the
Macaulay matrix $\Mac$, the vector ${\bf y}$ in its nullspace, the
moment matrix $\mathfrak{M}_S({\bf y})$, the polynomial $J$, the
matrix of traces $R$ and the (approximate) multiplication matrices
of the (approximate) radical, following Algorithm \ref{algo}.

The exact system:
\begin{scriptsize}
$$
{\bf f}=\begin{cases}
                 3x_1^2+18x_1x_2-48x_1+21x_2^2-114x_2+156\\
x_1^3-\frac{259}{4}x_1^2x_2+\frac{493}{4}x_1^2 -\frac{611}{4}x_1x_2^2+\frac{2423}{4}x_1x_2 -\frac{1175}{2}x_1\\
\quad\quad -5x_2^3+6x_2^2+x_2+5\\
x_1^3+\frac{81}{4}x_1^2x_2-\frac{163}{4}x_1^2+\frac{21}{4}x_1x_2^2+\frac{87}{4}x_1x_2-\frac{151}{2}x_1-x_2^3\\
\quad\quad +4x_2^2+2x_2+3
            \end{cases}
$$
\end{scriptsize}${\bf f}$ has common roots $(-1,3)$ of multiplicity 3 and $(2,2)$ of
multiplicity 2.

The system with clusters:
\begin{scriptsize}
$$
\bar{{\bf f}}=\begin{cases}
3x_1^2+17.4x_1x_2-46.5x_1+23.855x_2^2-127.977x_2+171.933\\
x_1^3-72.943x_1^2x_2+139.617x_1^2-8.417x_1x_2^2-124.161x_1x_2\\
\quad\quad +295.0283x_1-5x_2^3+6x_2^2+x_2+5\\
x_1^3+21.853x_1^2x_2-43.658x_1^2-27.011x_1x_2^2+185.548x_1x_2\\
\quad\quad -274.649x_1-x_2^3+4x_2^2+2x_2+3
            \end{cases}
$$
\end{scriptsize}$\bar{\bf f}$ has two clusters: $(-1,3),\, (-0.9,3),\, (-1.01,3.1) $ and $(2,2)$, $(1.9,2)$
 each of radius $10^{-1}$.

The perturbed system:
\begin{scriptsize}
$$
\hat{{\bf f}}=\begin{cases}
3x_1^2+18x_1x_2-48x_1+21.001x_2^2-113.999x_2+156.001\\
1.001x_1^3-64.751x_1^2x_2+123.250x_1^2-152.750x_1x_2^2\\
\quad \quad +605.751x_1x_2-587.500x_1-4.999x_2^3+6.0001x_2^2+x_2+5\\
x_1^3+20.249x_1^2x_2-40.750x_1^2+5.249x_1x_2^2+21.749x_1x_2 \\
\quad\quad -75.5x_1-1.001x_2^3+4x_2^2+2x_2+ 3
            \end{cases}
$$
\end{scriptsize}is obtained from ${\bf f}$ by a random perturbation of size
$10^{-3}$. This system has no common roots.

We set $\delta=6$, $D=2$ and $\Delta=6$. The Sylvester matrices in all three cases were size $28\times 28$ and
in the first two cases they had rank $23$ while in the last case it
was full rank. In the first two cases the fact that the corank is 5 indicates that there are 5 solutions, counting multiplicities. For these cases we computed a basis
$S:=[1,x_1,x_2,x_1x_2,x_1^2]$ for the factor algebra by taking
maximum rank submatrices of the Macaulay matrices. In the third
case, we simply erased the columns of the Macaulay matrix
corresponding to the monomials in $S$. From here, we chose random
elements in the nullspaces of the (cropped) Macaulay matrices to
compute the moment matrices: \begin{tiny}
$$
\left[\begin{array}{ccccc}
1& 0& 0& 0& 0\\
0& 0& 0& \frac{-6}{7}& \frac{-34}{7}\\
0& 0& \frac{-52}{7}& \frac{10}{7}& \frac{-6}{7}\\
0& \frac{-6}{7}& \frac{10}{7}& \frac{-40}{7}& \frac{-36}{7}\\
0& \frac{-34}{7}& \frac{-6}{7}& \frac{-36}{7}& \frac{-276}{7}\\
\end{array}\right],$$
$$
\left[\begin{array}{ccccc}
1& 0& 0& 0& 0\\
0& 0& 0& 1.787& -20.059\\
0& 0& -7.207& 1.219& 1.787\\
0& 1.787& 1.219& 7.702& -43.499\\
0& -20.059& 1.787& -43.499& -43.644\\
\end{array}\right]\,\text{ and}$$
$$
\left[\begin{array}{ccccc}
1& 0& 0& 0& 0\\
0& 0& 0& -0.858& -4.848\\
0& 0& -7.428& 1.428& -0.858\\
0& -0.858& 1.428& -5.719& -5.125\\
0& -4.848& -0.858& -5.125& -39.404\\
\end{array}\right].
$$
\end{tiny}The polynomials $J$, computed from the moment matrices are:
\begin{small}
$$
\begin{aligned}
J&=5-\frac{3}{10}x_1-\frac{26}{15}x_2-\frac{1}{30}x_1x_2-\frac{1}{5}x_1^2 \notag\\
\bar{J}&=5+ 0.916x_1 - 1.952x_2 - 0.636x_1 x_2 - 0.106x_1^2    \notag\\
\hat{J}&=4.999-0.306x_1-1.733x_2-0.030x_1x_2-0.200x_1^2 .\notag
\end{aligned}
$$
\end{small}

After computing the matrices ${\rm Syl}_{S}(J)$ and
$\mathfrak{R}_S({\bf y})$, we obtain the matrices of traces:

\begin{tiny}
$$
\left[\begin{array}{ccccc}
5& 1& 13& -1& 11\\
1& 11& -1& 25& 13\\
13& -1& 35& -11& 25\\
-1& 25& -11& 59& 23\\
11& 13& 25& 23& 35\\
\end{array}\right],$$
$$
\left[\begin{array}{ccccc} 4.999& 0.990& 13.100&
-1.031& 10.440\\
0.990& 10.440& -1.031&
23.812& 12.100\\
13.100& -1.031& 35.610&
-11.206& 23.812\\
-1.031& 23.812& -11.206&
56.533& 21.337\\
10.440& 12.100& 23.812& 21.337& 31.729\\
\end{array}\right]\,\text{ and}$$
$$
\left[\begin{array}{ccccc} 5& 0.995& 13.002&
-1.017& 11.003\\
0.995& 10.999& -1.015&
25.019& 12.913\\
13.002& -1.017& 35.013&
-11.064& 25.029\\
-1.017& 25.0256& -11.061&
59.129& 22.770\\
11.003& 12.870& 25.0519& 22.644& 34.968\\
\end{array}\right].$$

\end{tiny}

The first matrix $R$ has rank $2$, while $\bar{R}$ and $\hat{R}$
have rank 5. In the first case we follow steps 9 and 10 of Algorithm
\ref{algo} to obtain the multiplication matrices of the radical with respect to its basis $T=[1,\,x_1]$:
\begin{tiny}
$$
\left[\begin{array}{cc} 1&1\\
2&0
\end{array}\right]\;\text{ and }\left[\begin{array}{cc}
\frac{7}{3}&\frac{-1}{3}\\
\frac{-2}{3}&\frac{8}{3}
\end{array}\right],
$$
\end{tiny}with respective eigenvalues $[2,-1]$ and $[2,3]$.

For the second case we use the method described in
\cite{JaRoSza06,JaRoSza07} to compute the approximate multiplication
matrices of the approximate radical of the clusters. Using Gaussian
Elimination with complete pivoting, we found that the almost
vanishing pivot elements were of the order of magnitude of $10^{-1}$
which clearly indicated the numerical rank. Using the submatrices
obtained from the complete pivoting algorithm we got the following
approximate multiplication matrices of the approximate radical with respect to the basis $T=[x_1x_2,\,x_2]$:
\begin{tiny}
$$
\left[\begin{array}{cc} 0.976&1\\
1.895&-4.623\times 10^{-7}
\end{array}\right]\;\text{ and }\left[\begin{array}{cc}
2.346&-0.354\\
-0.671&2.691
\end{array}\right].
$$
\end{tiny}The norm of the commutator of these matrices is $0.002$ and their
eigenvalues are respectively $[1.949, -0.972]$ and $[2.001, 3.036]$. Note that the corresponding roots
$[1.949,2.001]$ and $[-0.972,3.036]$ are within $10^{-2}$ distance from the centers of gravity of the clusters,
as was shown in  \cite{JaRoSza06,JaRoSza07} (recall that the radius of the clusters was $10^{-1}$).

In the third case, the numerical rank was not easy to determine
using either SVD or complete pivoting. However, when we assume that the numerical
rank of $R$ is 2, and we cut the
matrix $R$ using the output of the complete pivoting algorithm,  then we obtain the multiplication matrices  with respect to the basis $T=[x_1x_2,\,x_2]$:
\begin{tiny}
$$
\left[\begin{array}{cc} 1.005&0.992\\
1.992&0.005
\end{array}\right]\;\text{ and }\left[\begin{array}{cc}
2.327&-0.330\\
-0.663&2.664
\end{array}\right].
$$
\end{tiny}The norm of the commutator of these matrices is $0.010$ and their
eigenvalues are respectively $[1.997, -0.987]$ and $[1.999, 2.993]$ (recall that the perturbation of the polynomials
was of size $10^{-3}$).

\section{Conclusion}

In this paper we gave an algorithm to compute matrices of traces and
the radical of an ideal $\I$ which has finitely many projective
common roots, none of them at infinity and its factor algebra is
Gorenstein. A follow-up paper will consider an extension of the
above algorithm which also works in the non-Gorenstein case and for
systems which have roots at infinity, as well as an alternative
method using Bezout matrices for the affine complete intersection
case to compute the radical $\sqrt{\I}$.

\bibliographystyle{abbrv}

{\small


\begin{thebibliography}{10}

\vspace{.3cm}
\bibitem{ArSo95}
I.~Armend\'{a}riz and P.~Solern\'{o}.
\newblock On the computation of the radical of polynomial complete intersection
  ideals.
\newblock In {\em AAECC-11: Proceedings of the 11th International Symposium on
  Applied Algebra, Algebraic Algorithms and Error-Correcting Codes}, pages
  106--119, 1995.

\bibitem{Be91}
E.~Becker.
\newblock Sums of squares and quadratic forms in real algebraic geometry.
\newblock In {\em De la g\'eom\'etrie alg\'ebrique r\'eelle (Paris, 1990)},
  volume~1 of {\em Cahiers S\'em. Hist. Math. S\'er. 2}, pages 41--57. 1991.

\bibitem{BeCaRoSz96}
E.~Becker, J.~P. Cardinal, M.-F. Roy, and Z.~Szafraniec.
\newblock Multivariate {B}ezoutians, {K}ronecker symbol and {E}isenbud-{L}evine
  formula.
\newblock In {\em Algorithms in algebraic geometry and applications (Santander,
  1994)}, volume 143 of {\em Progr. Math.}, pages 79--104.

\bibitem{BeWo94}
E.~Becker and T.~W{\"o}rmann.
\newblock On the trace formula for quadratic forms.
\newblock In {\em Recent advances in real algebraic geometry and quadratic
  forms}, volume 155 of {\em Contemp. Math.}, pages 271--291. 1994.

\bibitem{becwor96}
E.~Becker and T.~W\"{o}rmann.
\newblock Radical computations of zero-dimensional ideals and real root
  counting.
\newblock In {\em Selected papers presented at the international IMACS
  symposium on Symbolic computation, new trends and developments}, pages
  561--569, 1996.

\bibitem{brigonz02}
E.~Briand and L.~Gonzalez-Vega.
\newblock Multivariate {N}ewton sums: Identities and generating functions.
\newblock {\em Communications in Algebra}, 30(9):4527--4547, 2001.

\bibitem{carmou96}
J.~Cardinal and B.~Mourrain.
\newblock Algebraic approach of residues and applications.
\newblock In J.~Reneger, M.~Shub, and S.~Smale, editors, {\em Proceedings of
  AMS-Siam Summer Seminar on Math. of Numerical Analysis (Park City, Utah,
  1995)}, volume~32 of {\em Lectures in Applied Mathematics}, pages 189--219,
  1996.

\bibitem{CatDicStu96}
E.~Cattani, A.~Dickenstein, and B.~Sturmfels.
\newblock Computing multidimensional residues.
\newblock In {\em Algorithms in algebraic geometry and applications (Santander,
  1994)}, volume 143 of {\em Progr. Math.}, pages 135--164. 1996.

\bibitem{CatDicStu98}
E.~Cattani, A.~Dickenstein, and B.~Sturmfels.
\newblock Residues and resultants.
\newblock {\em J. Math. Sci. Univ. Tokyo}, 5(1):119--148, 1998.

\bibitem{Cox98}
D.~A. Cox, J.~B. Little, and D.~{O'Shea}.
\newblock {\em Using Algebraic Geometry}, volume 185 of {\em Graduate Texts in
  Mathematics}.
\newblock Springer-Verlag, NY, 1998.
\newblock 499 pages.

\bibitem{CuFi96}
R.~E. Curto and L.~A. Fialkow.
\newblock Solution of the truncated complex moment problem for flat data.
\newblock {\em Mem. Amer. Math. Soc.}, 119(568):x+52, 1996.

\bibitem{DaJe05}
C.~D'Andrea and G.~Jeronimo.
\newblock Rational formulas for traces in zero-dimensional algebras.
\newblock {\em http://arxiv.org/abs/math.AC/0503721}, 2005.

\bibitem{DiazGonz01}
G.~M. D{\'{\i}}az-Toca and L.~Gonz{\'a}lez-Vega.
\newblock An explicit description for the triangular decomposition of a
  zero-dimensional ideal through trace computations.
\newblock In {\em Symbolic computation: solving equations in algebra, geometry,
  and engineering (South Hadley, MA, 2000)}, volume 286 of {\em Contemp.
  Math.}, pages 21--35. 2001.

\bibitem{dickson}
L.~{Dickson}.
\newblock {\em Algebras and {T}heir {A}rithmetics}.
\newblock University of {C}hicago {P}ress, 1923.

\bibitem{ElMo07}
M.~{E}lkadi and B.~{M}ourrain.
\newblock {\em {I}ntroduction {\`a} la r{\'e}solution des syst{\`e}mes
  polynomiaux}, volume~59 of {\em {M}ath{\'e}matiques et {A}pplications}.
\newblock 2007.

\bibitem{GiMo89}
P.~Gianni and T.~Mora.
\newblock Algebraic solution of systems of polynomial equations using
  {G}r{\"o}bner bases.
\newblock In {\em Applied algebra, algebraic algorithms and error-correcting
  codes (Menorca, 1987)}, volume 356 of {\em Lecture Notes in Comput. Sci.},
  pages 247--257. 1989.

\bibitem{GiTrZa88}
P.~Gianni, B.~Trager, and G.~Zacharias.
\newblock Gr\"obner bases and primary decomposition of polynomial ideals.
\newblock {\em J. Symbolic Comput.}, 6(2-3):149--167, 1988.
\newblock Computational aspects of commutative algebra.

\bibitem{GonVeg}
L.~Gonz\'alez-Vega.
\newblock The computation of the radical for a zero dimensional ideal in a
  polynomial ring through the determination of the trace for its quotient
  algebra.
\newblock {\em Preprint}, 1994.

\bibitem{GoTr95}
L.~Gonz{\'a}lez-Vega and G.~Trujillo.
\newblock Using symmetric functions to describe the solution set of a
  zero-dimensional ideal.
\newblock In {\em Applied algebra, algebraic algorithms and error-correcting
  codes (Paris, 1995)}, volume 948 of {\em Lecture Notes in Comput. Sci.},
  pages 232--247.

\bibitem{GrPf02}
G.-M. Greuel and G.~Pfister.
\newblock {\em A {\bf {S}ingular} introduction to commutative algebra}.
\newblock 2002.
\newblock With contributions by Olaf Bachmann, Christoph Lossen and Hans
  Sch\"onemann, With 1 CD-ROM (Windows, Macintosh, and UNIX).

\bibitem{HeObPa06}
W.~Hei{\ss}, U.~Oberst, and F.~Pauer.
\newblock On inverse systems and squarefree decomposition of zero-dimensional
  polynomial ideals.
\newblock {\em J. Symbolic Comput.}, 41(3-4):261--284, 2006.

\bibitem{JaRoSza06}
I.~Janovitz-Freireich, L.~R\'{o}nyai, and \'{A}gnes Sz\'{a}nt\'{o}.
\newblock Approximate radical of ideals with clusters of roots.
\newblock In {\em ISSAC '06: Proceedings of the 2006 International Symposium on
  Symbolic and Algebraic Computation}, pages 146--153, 2006.

\bibitem{JaRoSza07}
I.~Janovitz-Freireich, L.~R\'{o}nyai, and \'{A}gnes Sz\'{a}nt\'{o}.
\newblock Approximate radical for clusters: a global approach using gaussian
  elimination or svd.
\newblock {\em Mathematics in Computer Science}, 1(2):393--425, 2007.

\bibitem{KoMoHo89}
H.~Kobayashi, S.~Moritsugu, and R.~W. Hogan.
\newblock On radical zero-dimensional ideals.
\newblock {\em J. Symbolic Comput.}, 8(6):545--552, 1989.

\bibitem{KrLo91}
T.~Krick and A.~Logar.
\newblock An algorithm for the computation of the radical of an ideal in the
  ring of polynomials.
\newblock In {\em Applied algebra, algebraic algorithms and error-correcting
  codes (New Orleans, LA, 1991)}, volume 539 of {\em Lecture Notes in Comput.
  Sci.}, pages 195--205.

\bibitem{KrLo91b}
T.~Krick and A.~Logar.
\newblock Membership problem, representation problem and the computation of the
  radical for one-dimensional ideals.
\newblock In {\em Effective methods in algebraic geometry (Castiglioncello,
  1990)}, volume~94 of {\em Progr. Math.}, pages 203--216. 1991.

\bibitem{Kun86}
E.~Kunz.
\newblock {\em K{\"a}hler differentials}.
\newblock Advanced lectures in {M}athematics. Friedr. Vieweg and Sohn, 1986.

\bibitem{Lak90}
Y.~N. Lakshman.
\newblock On the complexity of computing a {G}r{\"o}bner basis for the radical
  of a zero-dimensional ideal.
\newblock In {\em In Proceedings of the Twenty Second Symposium on Theory of
  Computing}, pages 555--563, 1990.

\bibitem{Lak91}
Y.~N. Lakshman.
\newblock A single exponential bound on the complexity of computing {G}r\"obner
  bases of zero-dimensional ideals.
\newblock In {\em Effective methods in algebraic geometry (Castiglioncello,
  1990)}, volume~94 of {\em Progr. Math.}, pages 227--234. 1991.

\bibitem{LakLaz91}
Y.~N. Lakshman and D.~Lazard.
\newblock On the complexity of zero-dimensional algebraic systems.
\newblock In {\em Effective methods in algebraic geometry (Castiglioncello,
  1990)}, volume~94 of {\em Progr. Math.}, pages 217--225. 1991.

\bibitem{LaLaRo07}
J.~B. Lasserre, M.~Laurent, and P.~Rostalski.
\newblock {A unified approach to computing real and complex zeros of
  zero-dimensional ideals}.
\newblock {\em preprint}, 2007.

\bibitem{LaLaRo072}
J.~B. Lasserre, M.~Laurent, and P.~Rostalski.
\newblock {Semidefinite characterization and computation of zero-dimensional
  real radical ideals}.
\newblock {\em to appear in Foundations of Computational Mathematics}, 2007.

\bibitem{La81}
D.~Lazard.
\newblock R\'esolution des syst\`emes d'\'equations alg\'ebriques.
\newblock {\em Theoret. Comput. Sci.}, 15(1):77--110, 1981.

\bibitem{BMr98}
B.~Mourrain.
\newblock Computing isolated polynomial roots by matrix methods.
\newblock {\em J. of Symbolic Computation, Special Issue on Symbolic-Numeric
  Algebra for Polynomials}, 26(6):715--738, Dec. 1998.

\bibitem{PeRoSz93}
P.~Pedersen, M.-F. Roy, and A.~Szpirglas.
\newblock Counting real zeros in the multivariate case.
\newblock In {\em Computational algebraic geometry (Nice, 1992)}, volume 109 of
  {\em Progr. Math.}, pages 203--224. Boston, MA, 1993.

\bibitem{Rouiller99}
F.~Rouiller.
\newblock Solving zero-dimensional systems through the rational univariate
  representation.
\newblock In {\em AAECC: Applicable Algebra in Engineering, Communication and
  Computing}, volume~9, pages 433--461, 1999.

\bibitem{ScSt75}
G.~Scheja and U.~Storch.
\newblock {\"U}ber {S}purfunktionen bei vollst{\"a}ndigen {D}urschnitten.
\newblock {\em J. Reine Angew Mathematik}, 278:174--190, 1975.

\bibitem{Stan96}
R.~P. Stanley.
\newblock {\em Combinatorics and commutative algebra}, volume~41 of {\em
  Progress in Mathematics}.
\newblock Birkh\"auser, 1996.

\bibitem{YoNoTa92}
K.~Yokoyama, M.~Noro, and T.~Takeshima.
\newblock Solutions of systems of algebraic equations and linear maps on
  residue class rings.
\newblock {\em J. Symbolic Comput.}, 14(4):399--417, 1992.

\bibitem{ReZh04}
L.~Zhi and G.~Reid.
\newblock Solving nonlinear polynomial systems via symbolic-numeric elimination
  method.
\newblock In {\em In Proceedings of the International Conference on Polynomial
  System Solving}, pages 50--53, 2004.

\end{thebibliography}
}

\end{document}